\documentclass[12pt,a4paper]{article}
\usepackage[a4paper,margin=1in]{geometry}
\usepackage[T1]{fontenc}
\usepackage[utf8]{inputenc}
\usepackage{lmodern}
\usepackage[expansion=false]{microtype}
\usepackage{setspace}
\usepackage{amsmath,amssymb,amsthm,mathtools}
\usepackage{booktabs,tabularx,array}
\usepackage{enumitem}
\usepackage[authoryear,round]{natbib}
\usepackage[dvipsnames]{xcolor}
\usepackage[colorlinks=true,linkcolor=MidnightBlue,citecolor=MidnightBlue,urlcolor=MidnightBlue,hyperfootnotes=false]{hyperref}
\usepackage{needspace}
\usepackage{sectsty}
\allsectionsfont{\raggedright}
\hypersetup{pdftitle={Learning to Interpret Forecasts: Evidence Quality and Source Familiarity},pdfauthor={Georgy Lukyanov and Nikita Ogorodnikov},pdfsubject={Dynamic evidence choice and learning about a forecasting source},pdfkeywords={forecasting services, information quality, learning, monopoly, calibration}}
\newtheorem{proposition}{Proposition}
\newtheorem{lemma}{Lemma}

\numberwithin{equation}{section}
\newcommand{\E}{\mathbb E}
\newcommand{\R}{\mathbb R}
\newcommand{\Vhat}{\widehat V}
\setlist[enumerate]{leftmargin=1.6em,itemsep=.4em}
\title{Learning to Interpret Forecasts:\\Evidence Quality and Source Familiarity}
\author{Georgy Lukyanov\thanks{Corresponding author. Toulouse School of Economics, 1 Esplanade de l'Universit\'e, 31080 Toulouse Cedex 06, France. Email: \href{mailto:georgy.lukyanov@tse-fr.eu}{georgy.lukyanov@tse-fr.eu}.}\and Nikita Ogorodnikov\thanks{International College of Economics and Finance, HSE University, 11 Pokrovsky Boulevard, building T, Moscow 109028, Russia. Email: \href{mailto:ngogorodnikov@edu.hse.ru}{ngogorodnikov@edu.hse.ru}.}}
\date{3 October 2026}
\begin{document}
\setstretch{1.1}
\maketitle

\begin{abstract}
We study how learning about a forecasting method affects a seller’s choice of evidence. More evidence improves the current forecast but makes the method’s persistent component harder to identify. A monopolist may therefore collect less evidence initially to help customers interpret later forecasts. Because the seller captures only part of customers’ gains from costly evidence, it switches to less initial evidence sooner than a welfare-maximizing regulator. A minimum evidence requirement improves welfare only for intermediate weights on future benefits. An independent assessment removes this learning trade-off but may be privately unprofitable despite being socially worthwhile.
\end{abstract}

\section{Introduction}\label{sec:intro}

When using a forecast, a customer needs to understand both the event being forecast and the source of the prediction. A firm deciding how much to produce may purchase a forecast of demand. If the forecasting method is unfamiliar, however, the firm may find it difficult to tell whether an optimistic prediction reflects favourable evidence or a persistent tendency of that method to produce high estimates. Experience with the same source helps the customer make this distinction. The usefulness of a forecast therefore depends not only on the evidence collected for the current decision, but also on what the customer has learned about the source.

This paper asks how that learning affects the evidence collected by a commercial forecasting service. In particular, can the prospect of future sales induce the service to produce a less informative first forecast, even when a more informative forecast would be profitable today? If so, when does a requirement to collect more evidence improve welfare?

Consider a forecasting method that combines a stable, source-specific component with evidence about the current event. When the evidence is weak, the method puts relatively more weight on its stable component. The resulting forecast is less useful for the current decision, but it is more revealing about the method itself. When the evidence is strong, the forecast follows the current event more closely. It then provides better guidance today, while revealing less about the persistent component that will also enter tomorrow's forecast.

To see the distinction, suppose that a new service issues a high demand forecast. If it has collected extensive evidence, the forecast is likely to reflect that evidence, and a customer learns relatively little about whether the service systematically produces high numbers. If the evidence is limited, the same prediction is more revealing about the service's persistent component. Once this component is better understood, customers can allow for it when interpreting a later forecast. Learning here means becoming better able to use the source's information. It does not mean becoming convinced that the source is more competent.

The service can benefit from this learning because customers are willing to pay more for a later forecast that they can interpret more accurately. Its first evidence choice consequently has three effects. Collecting more evidence raises current willingness to pay. It also raises production cost. Finally, it reduces the information that the first forecast conveys about the source, lowering the value of a later forecast produced by the same method. The last effect is the subject of the paper.

We study a two-period model in which a monopolist sells one forecast at each date. Each forecast concerns a new, independent state. The forecasting method has a persistent component that is initially unknown to both customers and the commercial seller. The seller chooses how much evidence to collect and the price of access, before the forecast is produced. A fixed rule then combines the evidence with the persistent component. The seller cannot rewrite the resulting forecast or separately disclose the component. The first forecast becomes public after the decision for which it was sold has expired. Thus everyone can learn from the old forecast before deciding whether to purchase the next one.

The central result concerns a choice between a basic and an intensive evidence format. We give conditions under which the intensive format is profitable in a one-period market and is also optimal in the second period after either possible first-period format. Nevertheless, a seller that attaches sufficient weight to future revenue chooses the basic format in the first period. It accepts a lower first price in return for customers understanding its method better and paying more for the second forecast. Both formats are available at both dates, at the same costs. The increase in evidence over time is therefore a choice of the seller rather than an assumed technological improvement.

This sequence is not automatically inefficient. Customers also benefit from learning how to interpret future forecasts. To evaluate an evidence requirement, we first consider a regulator that can require the intensive format but leaves the seller's pricing decisions in place. The seller values additional information through the revenue it generates; the regulator also counts the surplus received by customers. Because collecting more evidence has a real cost, the seller switches to the basic first format at a lower value of future revenue than the regulator would choose. There is an intermediate region in which the requirement improves welfare. When future benefits receive enough weight, however, both prefer the basic first format. Requiring intensive evidence in that region reduces welfare by preventing useful learning about the source.

We also consider an independent calibration assessment: an exercise that reveals the method's persistent component using information unavailable from the ordinary commercial forecast. With that component known, customers can remove it when interpreting each forecast. The service no longer needs to sacrifice current accuracy to make its later output easier to understand. Yet it may decline to pay for the assessment even when the improvement in customer decisions exceeds its cost. We separate this comparison with the unregulated market from the question of which intervention is best when an evidence requirement is also available.

The argument applies to a restricted class of forecasting technologies. The seller controls the evidence input to a fixed reporting rule, rather than choosing freely what to announce. Moreover, more evidence changes the weight of the persistent component in that rule. These restrictions matter: if the seller can disclose the component or the raw evidence at negligible cost, the reason for producing a less informative first forecast disappears. We use a forecasting service as a concrete interpretation of the model, without claiming that every forecasting firm, rating agency, or news outlet operates in this way.

\subsection*{Related literature}

The closest theoretical starting point is \citet{sethiyildiz2016}. An opinion in their model combines private information with an unobserved perspective. Repeated observation makes the perspective easier to identify, so the choice of whom to consult depends on both expertise and familiarity. Their Gaussian inference problem supplies the distinction between information about an event and information about its source that we use here. We make evidence precision a costly choice of a commercial supplier. The resulting contribution is a comparison of the evidence sequences preferred by the seller and by a regulator, taking account of how today's evidence choice affects the usefulness of tomorrow's forecast.

\citet{saghafiantomlinbiller2022} provide a closely related interpretation in terms of sensors whose outputs depend on inference models that other sensors do not fully know. Their analysis concerns which sources to consult and how communication links evolve as those models become familiar. \citet{sethiyildiz2026} examine how correlated perspectives within cultural groups affect communication. These papers establish the importance of understanding the process behind an observation. Our question concerns the incentives to change the evidence entering that process before the observation is produced.

The companion paper by \citet{lukyanovpototskiy2026} uses the same Gaussian learning technology in a different market problem. It takes source expertise and accumulated customer histories as given and studies renewal prices under competition, including the effects of outcome feedback and access to rival archives. Here a monopolist chooses evidence at both dates, thereby determining how much customers learn about its method. The shared learning structure is therefore an input to both papers; the present analysis concerns evidence provision, minimum evidence requirements, and the adoption of a separate calibration assessment.

The relation to \citet{campbell2015} concerns the public benefit from using information. His consumers choose between accessing an original source and relying on a freely available understanding of its contents that deteriorates over time. Access to the original replenishes that understanding. In our model, an archived forecast instead teaches customers how a persistent component enters a later forecast about a new state. The seller can improve current evidence while reducing this learning benefit. This distinction matters because our policy comparison concerns the production of evidence, rather than only the price of access to existing information.

Our pricing problem also relates to the literature on experience goods \citep{shapiro1983,goering1986,bergemannvalimaki2006,bonatti2011}. In \citet{bonatti2011}, for example, the quantity sold determines how much the market learns, giving the seller a reason to expand early sales through lower prices. Our public archive makes access to an old forecast independent of the initial purchase. The seller can therefore influence later learning through its evidence choice, while pricing at each date has a static form. Appendix~\ref{app:private} gives conditions under which that pricing result survives when only subscribers learn from the first forecast.

A separate comparison concerns incentives to distort reporting. \citet{gentzkowshapiro2006} and \citet{panova2026} study reporting motivated by reputation for competence. Panova's two-period setting, like ours, allows old reports to become public and affect subsequent demand, but the outlet chooses what to report after seeing a signal. In our model it chooses evidence before the signal arrives, and the reporting rule is fixed. Nor does improved familiarity mean a higher perceived probability of competence: customers become better able to correct for the persistent component. This also differs from learning about source accuracy through potentially biased reasoning, as in \citet{gentzkowwongzhang2025}.

\citet{kawamuralequement2023} identify another reason to reduce current informativeness: an outlet may misreport a persistent state to keep the public uncertain and encourage future consultation. Here states are independent, and the first forecast raises future demand by reducing uncertainty about the source. Finally, \citet{galpertitrevino2020} study the accuracy and clarity supplied in competition for attention, while \citet{jannschottmuller2026} show how less informative reporting can differentiate competing outlets and soften price competition. Our mechanism requires neither coordination among customers nor competing suppliers. It depends on how the evidence input changes the relative weight of a persistent component in a forecast.

Section~\ref{sec:model} describes the market. Section~\ref{sec:learning} derives the two kinds of information conveyed by a forecast. Section~\ref{sec:choice} solves the seller's pricing and evidence decisions. Sections~\ref{sec:welfare} and \ref{sec:audit} consider evidence requirements and calibration assessments. Section~\ref{sec:scope} discusses the scope of the mechanism. The appendices give the calculations and extensions to general customer distributions, continuous evidence choices, private histories, state feedback, and imperfect persistence.

\section{The market and the forecasting method}\label{sec:model}

\subsection{Customers and the value of a forecast}

There are two periods, a single forecasting service, and a unit mass of customers. In period $t\in\{1,2\}$, customers face a state $\theta_t$ drawn from $N(0,1)$. The two states are independent. For the demand-forecast interpretation, $\theta_t$ is demand relative to a known normal level. Customers choose their actions before observing the state.

A customer with stakes $\gamma$ chooses $a_t\in\R$ and has gross payoff
\begin{equation}\label{eq:payoff}
 -\gamma(a_t-\theta_t)^2.
\end{equation}
The parameter $\gamma$ measures how costly a forecasting error is for that customer. It is fixed across periods and uniformly distributed on $[0,1]$. A payment for the forecast is subtracted from this payoff. We use the uniform distribution in the main text so that the pricing and welfare calculations can be followed directly. Appendix~\ref{app:general} shows that the results extend to general stakes distributions.

Without a current forecast, the optimal action is the prior mean, zero, and the expected squared error is one. After purchasing a forecast, the customer chooses the conditional mean of the state. If the forecast reduces the expected squared error by $z$, it is worth $\gamma z$ to this customer. Throughout the paper, the \emph{value of a forecast} means this reduction for a customer with $\gamma=1$. It is a measure of its usefulness for decisions, rather than the accuracy of the published number before customers have interpreted it.

\subsection{Evidence and a persistent component}

The service operates a fixed forecasting method. The method has a persistent component
\begin{equation}\label{eq:prior}
 b\sim N(\mu,\eta^{-1}),\qquad \eta>0,
\end{equation}
independent of the two states. The number $b$ is the same at both dates. It represents an unknown tendency of the method to pull its forecasts towards a particular level. The parameter $\eta$ is the precision of initial knowledge about that tendency: a larger $\eta$ means less uncertainty about $b$.

Neither customers nor the commercial seller observe $b$ directly. It is a fixed effect of the reporting technology, not a privately known preference of a strategic forecaster. This distinction will prevent the seller from simply announcing $b$. The prior mean $\mu$ is common knowledge; uncertainty about the realised component is what makes the method difficult to interpret.

Before each forecast is produced, the seller selects evidence precision $q_t$. The method receives an internal signal
\begin{equation}\label{eq:signal}
 s_t=\theta_t+\varepsilon_t,\qquad \varepsilon_t\sim N(0,q_t^{-1}),
\end{equation}
with independent errors, and produces the reported forecast
\begin{equation}\label{eq:forecast}
 m_t=\frac{1}{1+q_t}b+\frac{q_t}{1+q_t}s_t.
\end{equation}
Thus higher $q_t$ has two linked effects: the evidence becomes less noisy, and the method puts more weight on that evidence. Its weight on the persistent component falls. Equation~\eqref{eq:forecast} is a technological assumption. It has the form of a Gaussian posterior mean based on a unit-precision prior with mean $b$, but the seller is not a Bayesian adviser who knows $b$ and chooses a message after seeing $s_t$.

For example, at $q_t=1/2$ the forecast is $2b/3+s_t/3$. At $q_t=1$ it is $(b+s_t)/2$, and the variance of the evidence noise is also halved. Comparing the two formats therefore involves both more reliable evidence and less reliance on the method's persistent component. Section~\ref{sec:scope} explains why the result need not hold if one changes only the noise in the evidence while keeping the loading on $b$ fixed.

In the main model the seller chooses between two formats,
\begin{equation}\label{eq:formats}
 q_t\in\{q_L,q_H\},\qquad 0<q_L<q_H.
\end{equation}
The basic format $q_L$ has cost zero and the intensive format $q_H$ has cost $c>0$ per forecast. These are incremental production costs; common operating costs can be added without changing the format comparisons provided that the service remains active. The same formats and costs are available in both periods. Distribution to an additional customer is costless.

Evidence precision is observable and contractible. One interpretation is a verified choice between two sampling or research protocols. The fixed technology releases only the resulting forecast. It does not release the raw signal to customers or to the commercial seller. Otherwise, observing both $s_t$ and $m_t$ would reveal $b=(1+q_t)m_t-q_ts_t$ exactly. Nor can the seller obtain arbitrary extra outputs at zero cost. The benchmark permits one commercial forecast per decision; a separate calibration exercise is introduced explicitly in Section~\ref{sec:audit}.

\subsection{Timing and information}

The first forecast is useful only before the first decision. Once that decision has been taken, the service places the forecast in a public archive. An old demand forecast, for example, can no longer improve an inventory decision that has already been made, while still providing information about the forecasting method. We assume initially that the first state and customers' realised payoff feedback do not become observable before the second sale. Appendix~\ref{app:feedback} allows additional information about the first state.

The order of events is as follows.
\begin{enumerate}
\item The seller announces $q_1$ and $p_1$. Customers decide whether to buy. The method receives $s_1$ and produces $m_1$. Subscribers observe $m_1$ and choose their actions.
\item After the first decision has expired, $m_1$ is released publicly. Customers and the seller update their beliefs about $b$.
\item The seller announces $q_2$ and $p_2$. Customers decide whether to buy. The method receives $s_2$ and produces $m_2$. Subscribers observe $m_2$ and choose their second actions.
\end{enumerate}

The seller discounts second-period profits by $\beta\in[0,1]$. We use the same discount factor when comparing social welfare across dates. Customers may be forward-looking, but an individual customer's initial purchase has no effect on access to the public archive. It therefore has no effect on that customer's second-period opportunities. The seller is a monopolist throughout, and customers cannot resell a live forecast.

The equilibrium is obtained by backward induction. At each date customers purchase when their expected benefit covers the price and use the forecast to choose their conditional-mean action. The seller chooses the price and format that maximise its remaining profit. These choices precede the current signal, and the seller has no private observation of $b$. Prices and formats consequently do not carry an additional private signal about the component.

The first forecast changes both the mean and the precision of beliefs about $b$. Its realised value affects the correction that customers make when using the second forecast. In the normal-quadratic model, however, the reduction in expected squared error depends only on the precision of those beliefs. That precision depends on $q_1$, but not on the realised number $m_1$. This is why the second-period price and format can be solved as functions of the first format alone.

\section{What customers learn from a forecast}\label{sec:learning}

\subsection{Information about the current decision}

Suppose that, before a forecast, the common belief about $b$ is normal with mean $\mu_x$ and precision $x$. Initially $x=\eta$. A customer can rescale the reported forecast and subtract its known mean component to obtain
\begin{equation}\label{eq:rescale}
 \frac{1+q}{q}m-\frac{\mu_x}{q}
 =\theta+\varepsilon+\frac{b-\mu_x}{q}.
\end{equation}
The right-hand side is a signal of the current state. It contains two sources of noise. The evidence itself has variance $1/q$. Uncertainty about the persistent component adds variance $1/(q^2x)$. Since these noises are independent, their variances add.

A normal signal with noise variance $n$ reduces the variance of a unit-variance state from one to $n/(1+n)$. Applying this calculation to \eqref{eq:rescale}, the value of a forecast is
\begin{equation}\label{eq:V}
 V(x,q)=\frac{1}{1+1/q+1/(q^2x)}
 =\frac{q^2x}{1+xq+xq^2}.
\end{equation}
The formula gives two straightforward comparisons. Holding customers' knowledge of the source fixed, more evidence raises the value of the current forecast. Holding the evidence fixed, a better-understood source also provides a more useful forecast. In the first comparison the current signal improves; in the second, customers can correct more accurately for its persistent component. Appendix~\ref{app:calculations} gives the derivatives and the optimal action explicitly.

\subsection{Information about the source}

Now consider the different question of what the same observation reveals about $b$. Multiplying the forecast by $1+q$ gives
\begin{equation}\label{eq:bsignal}
 (1+q)m=b+q\theta+q\varepsilon.
\end{equation}
For learning about the source, $b$ is the object of interest and the current state is noise. The variance of that noise is $q^2+q$: the first term comes from variation in the current state, and the second from the evidence error. The additional precision about $b$ supplied by the forecast is therefore
\begin{equation}\label{eq:lambda}
 \lambda(q)=\frac{1}{q(1+q)}.
\end{equation}
After seeing it, customers' precision about the source rises from $x$ to $x+\lambda(q)$.

Unlike current forecast value, this precision increment is decreasing in $q$. A forecast produced with intensive evidence is more strongly influenced by the current state. As an observation of the persistent component, it is therefore noisier. This is compatible with it being a better signal of the state: the two inference problems have different objects of interest. We are not adding random errors to an otherwise fixed forecast to make it useful. We are comparing two evidence inputs that change how the forecast combines a current signal with a persistent component.

\begin{lemma}\label{lem:learning}
The value $V(x,q)$ is strictly increasing in both $x$ and $q$. The precision increment $\lambda(q)$ is strictly decreasing in $q$. Consequently, a basic first forecast is less useful for the first decision but makes any given second forecast more useful than an intensive first forecast would.
\end{lemma}

The last statement uses independence of the two states. Learning about $\theta_1$ does not help customers predict $\theta_2$. What carries over is knowledge of the common component $b$. The first forecast becomes public only after the first decision, so all potential second-period customers receive this learning benefit.

\subsection{A numerical example}\label{sec:example}

Let $\eta=0.1$, $q_L=0.5$, and $q_H=1$. These values describe an initially poorly understood source and two available evidence formats. For the moment, compare the first formats while holding the second format at $q_H$. We will verify below that the seller actually chooses that second format.

\begin{table}[htbp]
\centering
\caption{A better current forecast can teach less about its source}\label{tab:information}
\begin{tabularx}{\textwidth}{@{}Xrr@{}}
\toprule
& Basic first format & Intensive first format\\
& $q_1=0.5$ & $q_1=1$\\
\midrule
Value for the first decision, $V(\eta,q_1)$ & 0.02326 & 0.08333\\
Additional precision about $b$, $\lambda(q_1)$ & 1.33333 & 0.50000\\
Precision about $b$ before period 2 & 1.43333 & 0.60000\\
Value of the second forecast at $q_2=1$ & 0.37069 & 0.27273\\
\bottomrule
\end{tabularx}
\smallskip
\begin{minipage}{\textwidth}\small
\emph{Notes:} Forecast value is the reduction in expected squared error for a customer with unit stakes. The last row uses $V(\eta+\lambda(q_1),1)$. The example is illustrative and is not an empirical calibration.
\end{minipage}
\end{table}

Intensive evidence raises the first forecast's value by about $0.06008$. Basic evidence, however, teaches more about $b$, and this raises the value of the same second forecast by about $0.09796$. The firm will compare the revenue associated with these two changes, allowing for the evidence cost and for how much it values the second period. The regulator will make a similar comparison using the benefits to customers as well as the firm's revenue.

\section{The seller's evidence choice}\label{sec:choice}

We solve the pricing problem first, then the second-period format choice, and finally the first-period choice. This order makes explicit why the incentive to use less evidence initially comes from future forecast value.

\subsection{Pricing at either date}\label{sec:pricing}

If a forecast has value $z>0$, a customer buys at price $p$ when $\gamma z\ge p$. Uniform stakes imply demand $1-p/z$ for $p\in[0,z]$. The seller's revenue is
\begin{equation}\label{eq:revenue}
 p\left(1-\frac{p}{z}\right).
\end{equation}
It is maximised at $p=z/2$, giving revenue $z/4$. Thus customers with $\gamma\ge1/2$ buy at each date. Producing a more valuable forecast raises the optimal price in the same proportion, leaving the number of buyers unchanged.

This calculation applies in the first period as well as the second. Customers who do not buy initially still observe the archive before the second sale. Buying early gives access to useful current information, but no additional opportunity to learn about $b$ for future use. The seller therefore charges the static monopoly price for the first forecast. A lower initial price under the basic format reflects its lower current value; it is not an extra discount intended to expand the initial audience.

For first and second formats $(q_1,q_2)$, discounted profit is consequently
\begin{equation}\label{eq:profit}
 \Pi(q_1,q_2)
 =\frac14 V(\eta,q_1)-C(q_1)
 +\beta\left[\frac14 V(\eta+\lambda(q_1),q_2)-C(q_2)\right],
\end{equation}
where $C(q_L)=0$ and $C(q_H)=c$. With the common operating cost normalised to zero, the basic format has positive profit, so the seller is active at each date.

\subsection{Why the second forecast uses intensive evidence}

At the second date there is no later sale for which the seller needs to make its method better understood. If the current precision about $b$ is $x$, intensive evidence is strictly preferred to basic evidence exactly when
\begin{equation}\label{eq:secondchoice}
 \frac14\bigl[V(x,q_H)-V(x,q_L)\bigr]>c.
\end{equation}
The left-hand side is the increase in subscription revenue and the right-hand side is the additional production cost.

Write
\begin{equation}\label{eq:xLH}
 x_L=\eta+\lambda(q_L),\qquad x_H=\eta+\lambda(q_H),
\end{equation}
so $x_L>x_H>\eta$. We focus first on the case in which \eqref{eq:secondchoice} holds at both $x_L$ and $x_H$. This is a checkable condition on the primitives. It ensures that the second format is intensive after either initial choice. We also require intensive evidence to be the profitable choice in a one-period market at initial precision $\eta$.

In the example from Table~\ref{tab:information}, take $c=0.008$. The intensive-minus-basic value gain in the second period is $0.19800$ following basic evidence and $0.16928$ following intensive evidence. A quarter of each number exceeds $0.008$, so the seller chooses intensive evidence after either history. Initially, the corresponding revenue gain is $0.06008/4=0.01502$, which also exceeds the cost. A seller interested only in the first period would choose intensive evidence.

\subsection{Why the first forecast may use basic evidence}

There are now two candidate sequences: intensive evidence at both dates, or basic evidence followed by intensive evidence. Define the current value lost by using the basic first format as
\begin{equation}\label{eq:DV}
 \Delta V=V(\eta,q_H)-V(\eta,q_L)>0,
\end{equation}
and the resulting gain in the second forecast's value as
\begin{equation}\label{eq:DH}
 \Delta H=V(x_L,q_H)-V(x_H,q_H)>0.
\end{equation}
These differences refer to two alternative first-period choices, not to the change in value between dates along a single sequence.

Switching the first format from intensive to basic reduces current revenue by $\Delta V/4$, saves cost $c$, and raises discounted continuation revenue by $\beta\Delta H/4$. The second-period production cost is unchanged. Hence
\begin{equation}\label{eq:profitdiff}
 \Pi(q_L,q_H)-\Pi(q_H,q_H)
 =-\frac{\Delta V}{4}+c+\beta\frac{\Delta H}{4}.
\end{equation}
This single comparison determines the seller's initial choice.

\begin{proposition}\label{prop:sequence}
Suppose that intensive evidence is strictly optimal at date 2 after either first format, that $\Delta V>4c$, and that
\begin{equation}\label{eq:interior}
 \Delta V-4c<\Delta H.
\end{equation}
Then the threshold
\begin{equation}\label{eq:betaP}
 \beta_P=\frac{\Delta V-4c}{\Delta H}
\end{equation}
lies in $(0,1)$. The seller chooses $(q_H,q_H)$ for $\beta<\beta_P$ and $(q_L,q_H)$ for $\beta>\beta_P$. Both sequences are optimal at equality. In both cases, half of the customers buy at each date.
\end{proposition}

The first two assumptions exclude a simple cost-cutting explanation: intensive evidence is worthwhile in isolation and remains worthwhile at the last date. The seller reduces first-period evidence because that choice improves the value of its future product. Condition~\eqref{eq:interior} ensures that this future gain can outweigh the current net revenue loss for some admissible discount factors. Without it, the seller may prefer intensive evidence at both dates for every $\beta\le1$.

In the numerical example, $\beta_P=0.28662$. At $\beta=0.34$, the seller chooses basic evidence first and intensive evidence second. Its prices are then $0.01163$ and $0.18534$. If it instead chose intensive evidence at both dates, the prices would be $0.04167$ and $0.13636$. It sacrifices current revenue in exchange for higher subsequent revenue, without changing the initial customer share.

The binary menu makes the result easy to state but is not essential to the downward pressure on initial evidence. Appendix~\ref{app:continuous} allows a continuous evidence choice. Greater weight on future profits weakly lowers the optimal first-period input because the value of learning about the source is decreasing in that input. A separate complementarity condition is needed to conclude that the second-period input is higher than the first. We do not claim that evidence must increase over time for every forecasting technology or cost function.

\section{When does an evidence requirement help?}\label{sec:welfare}

The seller's decision does not by itself tell us whether the basic first forecast is undesirable. Customers lose information about the first state, but they gain information that helps them interpret the second forecast. A welfare comparison must count both effects and the real cost of collecting evidence.

We begin with a regulator that can require the intensive evidence format in the first period. The seller continues to set monopoly prices, and customers remain free to purchase. The second-period format is intensive under either first-period choice, as in Proposition~\ref{prop:sequence}. This is a comparison of evidence policies holding the market's pricing institution fixed.

For a forecast with value $z$, the gross benefit to the customers who purchase is
\begin{equation}\label{eq:gross}
 \int_{1/2}^1\gamma z\,d\gamma=\frac38z.
\end{equation}
Subscription payments transfer money from customers to the seller and cancel from total welfare. Subtracting the seller's revenue $z/4$ from the gross benefit gives consumer surplus $z/8$. A unit increase in forecast value therefore raises revenue by $1/4$ but raises total gross benefit by $3/8$.

Relative to an intensive first format, basic evidence changes discounted welfare by
\begin{equation}\label{eq:welfarediff}
 W(q_L,q_H)-W(q_H,q_H)
 =-\frac38\Delta V+c+\beta\frac38\Delta H.
\end{equation}
The terms are, respectively, the loss in first-period decision benefits, the cost saving, and the gain in second-period decision benefits. Comparing \eqref{eq:welfarediff} with \eqref{eq:profitdiff} shows exactly where the private and social calculations differ: the regulator gives greater weight to both information effects, while the real cost is the same.

\begin{proposition}\label{prop:standard}
Under Proposition~\ref{prop:sequence}, suppose also that
\begin{equation}\label{eq:betaS}
 \beta_S=\frac{\Delta V-(8/3)c}{\Delta H}<1.
\end{equation}
Then $0<\beta_P<\beta_S<1$. For $\beta<\beta_P$, the seller already chooses intensive evidence. For $\beta\in(\beta_P,\beta_S)$, requiring intensive first-period evidence raises welfare. For $\beta>\beta_S$, the requirement lowers welfare. At $\beta_S$ the two first formats give equal welfare.
\end{proposition}

The seller switches earlier because it does not receive all the benefit that customers obtain from the extra expenditure on evidence. At the private switching point, the intensive-first sequence still provides more total discounted information value. Its additional revenue just covers its additional cost. Customers' additional surplus makes that sequence socially preferable. As the weight on the future rises further, the benefit of learning more about the source eventually outweighs this advantage, even from the regulator's perspective.

This reasoning also explains the role of cost. The difference between the thresholds is
\begin{equation}\label{eq:betagap}
 \beta_S-\beta_P=\frac{4c}{3\Delta H}>0.
\end{equation}
If the two formats had the same cost, the seller and the regulator would rank their discounted information values in the same way. Monopoly pricing alone does not create a disagreement over this format comparison when all value changes are scaled by the same revenue coefficient and there is no cost difference.

In the numerical example, $\beta_S=0.39550$. At $\beta=0.34$, the seller's basic-first sequence is profitable but a requirement to use intensive evidence raises welfare. Above $0.39550$, the same requirement is harmful. These numbers locate the three cases within the example; they are not estimates of a regulatory threshold in an actual forecasting market.

The intervention considered here regulates an input that is observable in the model, such as which evidence-collection protocol is used. It does not prescribe the forecast's realised conclusion. It also leaves the access distortion from monopoly pricing in place. Appendix~\ref{app:general} compares the seller with a planner that can give the forecast to everyone; the latter switches to basic evidence still later under the same conditions. Neither comparison makes low initial evidence universally undesirable.

\section{An independent assessment of the method}\label{sec:audit}

An evidence requirement improves the current forecast at the cost of reducing what it teaches about the source. A different intervention can address the initial uncertainty about the source directly. Suppose that, before the first sale, the seller can pay $K\ge0$ for an independent calibration assessment that reveals $b$ publicly. The assessment uses benchmark cases or privileged information that cannot be obtained from the single commercial output. Exact revelation is an ideal benchmark; a finite sample of noisy cases would generally provide only partial information.

Once $b$ is known, a customer can subtract its contribution from \eqref{eq:forecast} and recover $s_t$. The forecast's value becomes
\begin{equation}\label{eq:Vhat}
 \Vhat(q)=\frac{q}{1+q}.
\end{equation}
There is no further source component to learn about. The seller therefore chooses the same format at both dates. It chooses the intensive format strictly when
\begin{equation}\label{eq:audithigh}
 \frac14\bigl[\Vhat(q_H)-\Vhat(q_L)\bigr]>c.
\end{equation}
This condition must be checked separately. Removing calibration uncertainty need not always induce more evidence because it changes the marginal value of the input as well as its level of usefulness.

Consider parameters for which the unregulated seller chooses basic evidence followed by intensive evidence and \eqref{eq:audithigh} holds. The assessment then changes the sequence to intensive evidence at both dates. Let
\begin{equation}\label{eq:Z}
 Z=\Vhat(q_H)-V(\eta,q_L)
 +\beta\bigl[\Vhat(q_H)-V(x_L,q_H)\bigr]>0
\end{equation}
be the increase in total discounted forecast value. First-period production cost rises by $c$; the second-period production cost is unchanged.

\begin{proposition}\label{prop:audit}
Under the preceding conditions, the seller is willing to buy the assessment exactly when
\begin{equation}\label{eq:KP}
 K\le K_P=\frac14Z-c.
\end{equation}
Relative to the unaudited, unregulated market, introducing the assessment raises welfare exactly when
\begin{equation}\label{eq:KS}
 K<K_S=\frac38Z-c,
\end{equation}
with equality at $K_S$. In particular,
\begin{equation}\label{eq:Kgap}
 K_S-K_P=\frac18Z>0.
\end{equation}
\end{proposition}

The proof is a direct accounting comparison. Revenue rises by $Z/4$, while customers' gross benefit rises by $3Z/8$. Both calculations subtract the additional evidence cost and the assessment cost. The difference between the two gains is the increase in consumer surplus, $Z/8$. For $K\in(K_P,K_S)$ the seller declines an assessment that would improve welfare relative to the unregulated outcome.

This is not a claim that assessment is always the best available intervention. For $\beta\in(\beta_P,\beta_S)$, a first-period evidence requirement also improves welfare without paying $K$. If both instruments are available, the regulator should compare the assessment with that improved unaudited outcome. Appendix~\ref{app:assessment} gives the adjusted cost threshold. Likewise, a subsidy would have to account for any cost of public funds. The result identifies a benefit that the seller fails to capture; it does not dispense with the comparison of policy costs.

At $\beta=0.34$ in the numerical example, $Z=0.52071$, $K_P=0.12218$, and $K_S=0.18727$. The assessment cost is a one-off expenditure, while $c=0.008$ is the cost of using intensive evidence on a particular forecast. These are different activities and there is no reason for their costs to be similar.

The assessment also helps explain the limits of the baseline. At zero assessment cost the seller can do better by revealing the component, even if it were to keep its original evidence choices. Hence free, feasible calibration removes the commercial reason to remain difficult to interpret. A noisy assessment instead raises initial precision from $\eta$ to $\eta+\kappa$, where $\kappa$ is its precision. It improves each fixed sequence of forecasts, but its effect on the chosen evidence sequence is not globally monotone. Exact revelation and condition~\eqref{eq:audithigh} provide a case in which the direction of the format change is unambiguous.

\section{Scope and interpretation}\label{sec:scope}

\subsection{Why the reporting rule matters}

The key restriction is that additional evidence crowds out the persistent component in the reported forecast. The method in \eqref{eq:forecast} gives less weight to $b$ as $q$ rises. An observation of the forecast then tells customers more about the current event and less about the source. This provides an incentive to vary evidence over time.

It would be incorrect to claim the same effect for every biased forecast. Suppose, for comparison, that the reporting rule were $m=\theta+b+\varepsilon$, with noise variance $1/q$. The additional precision about $b$ from observing $m$ would be $q/(1+q)$, which \emph{increases} with $q$. More precise evidence would improve both the current decision and learning about the source. The trade-off in this paper would be absent. Thus persistence and unfamiliarity alone do not imply the result; the way in which the method combines current evidence and its persistent component is essential.

\subsection{Why the seller cannot simply explain the method}

The model distinguishes knowledge of the reporting rule from knowledge of its realised fixed effect. Customers know the functional form, the chosen precision, and the prior distribution of $b$. What they lack is the realised value of $b$. The commercial seller lacks that value too. It cannot resolve the uncertainty through a costless statement of something it knows.

This distinction is demanding. It can describe a seller that operates an inherited or externally supplied forecasting process whose persistent behaviour must be estimated from its outputs, while its research protocol can be verified. It does not describe a forecaster who knows a numerical bias and freely chooses a message. The latter may have a mutually beneficial way to reveal the bias. Similarly, if a seller observes the raw signal and the reported forecast, it can calculate $b$; if it can release enough inexpensive benchmark outputs, it can help customers estimate $b$ separately. These opportunities have to be costly or technologically unavailable for the baseline to be relevant. The calibration exercise makes one such opportunity explicit rather than assuming that opacity is desirable even when it is free to remove.

\subsection{What would distinguish the mechanism empirically?}

An increase in prices or accuracy over the life of a service would not by itself identify the mechanism. Cost changes, reputation, and many other forces can produce these patterns. A more specific implication concerns how an early forecast affects the use of a later one.

One could expose otherwise comparable decision makers to a first forecast produced under basic or intensive evidence by the same persistent reporting process. Everyone would then receive a second forecast produced with the same intensive format. The model predicts better second-period decisions following the basic first format, despite worse first-period decisions. Providing an independent estimate of the persistent component should reduce this advantage; revealing the component exactly should eliminate it. Such an exercise would test the learning channel. The present paper does not provide evidence that the required technological restrictions hold in a particular industry.

Additional feedback about the first state does not reverse the learning comparison within our reporting rule. It helps customers separate that state from the persistent component, but lower first-period evidence still reveals more about the latter. Imperfect persistence weakens the connection between learning about the first and second forecasts; if the source component is independently redrawn, that connection disappears. Appendix~\ref{app:feedback} derives these statements. They concern the information mechanism, rather than establishing that every policy threshold is unchanged by feedback or drift.

\section{Conclusion}\label{sec:conclusion}

A customer may learn two things from the same forecast: something about the current event, and something about how the forecasting method works. In the technology studied here, collecting more evidence improves the first kind of learning and reduces the second. A seller anticipating future sales may therefore choose a basic first forecast even though intensive evidence would be profitable in a one-period market. Once the method has become better understood, it chooses intensive evidence for the later forecast.

This choice can be either socially useful or excessive. The comparison depends on the future value of understanding the source and on the cost of intensive evidence. A minimum evidence requirement helps when the seller switches too early. It harms welfare when the gains from learning about the source, together with the saving in evidence cost, outweigh the loss of information for the first decision. An independent assessment of the method offers a different response: it can make the source understandable without sacrificing the current evidence input. Whether this response is worthwhile depends on the assessment cost and on the other interventions available.

\appendix
\section{Calculations and proofs for the main model}\label{app:calculations}

\subsection{Beliefs, decisions, and the two learning effects}

Let the belief before the current forecast be $b\sim N(\mu_x,x^{-1})$. The centred and rescaled observation in \eqref{eq:rescale} is $y=\theta+\zeta$, where $\zeta$ is independent of $\theta$ and has variance
\begin{equation}\label{eq:noise}
 n(x,q)=\frac1q+\frac{1}{q^2x}.
\end{equation}
Normal updating gives $\E[\theta\mid y]=y/(1+n)$ and $\operatorname{Var}(\theta\mid y)=n/(1+n)$. Thus the customer's optimal action is
\begin{equation}
 a^*(m)=V(x,q)\left(\frac{1+q}{q}m-\frac{\mu_x}{q}\right),
\end{equation}
and the reduction in expected loss is $\gamma V(x,q)$, with $V$ given by \eqref{eq:V}. Its derivatives are
\begin{align}
 V_q(x,q)&=\frac{xq(2+xq)}{[1+xq(1+q)]^2}>0,\label{eq:Vq}\\
 V_x(x,q)&=\frac{q^2}{[1+xq(1+q)]^2}>0.\label{eq:Vx}
\end{align}
For learning about the source, \eqref{eq:bsignal} is an observation of $b$ with independent normal noise of variance $q^2+q$. Its precision is $\lambda(q)$ from \eqref{eq:lambda}. The posterior mean and precision are
\begin{equation}
 \mu_x^+=\frac{x\mu_x+\lambda(q)(1+q)m}{x+\lambda(q)},
 \qquad x^+=x+\lambda(q).
\end{equation}
The updated mean varies with the report; the updated precision does not. Finally,
\begin{equation}
 \lambda'(q)=-\frac{1+2q}{q^2(1+q)^2}<0.
\end{equation}
Since $V_x>0$, a lower first input raises $V(\eta+\lambda(q_1),q_2)$ at any fixed $q_2$. This proves Lemma~\ref{lem:learning}.

\subsection{Pricing and the evidence sequence}

At a date with value $z$, prices outside $[0,z]$ cannot improve on positive revenue within that interval. On $[0,z]$, revenue has derivative $1-2p/z$ and second derivative $-2/z<0$. Hence its unique optimum is $p=z/2$ with demand $1/2$. The public release makes the customer's continuation payoff independent of the initial purchase, so this calculation applies at both dates.

At date 2, subtracting the profits under the two formats gives $[V(x,q_H)-V(x,q_L)]/4-c$. The assumptions of Proposition~\ref{prop:sequence} make this expression positive after either first format. Given that continuation, the initial comparison is \eqref{eq:profitdiff}. The coefficient of $\beta$ is positive. The assumptions $\Delta V>4c$ and $\Delta V-4c<\Delta H$ place its zero strictly between zero and one and give \eqref{eq:betaP}. The signs on either side of that zero prove the proposition. This argument also establishes sequential optimality of the second-period choice, rather than requiring a commitment to the whole sequence.

For reference, a sufficient condition for intensive evidence to remain optimal after either history is
\begin{equation}\label{eq:sufficient}
 \bigl[\eta+\lambda(q_L)\bigr]q_Lq_H<1,
 \qquad \Delta V>4c.
\end{equation}
To see this, define $D(x)=V(x,q_H)-V(x,q_L)$. From \eqref{eq:Vx}, the sign of $D'(x)$ is the sign of
\begin{align}
 q_H[1+xq_L(1+q_L)]-q_L[1+xq_H(1+q_H)]
 =(q_H-q_L)(1-xq_Lq_H).
\end{align}
Under \eqref{eq:sufficient}, $D$ is strictly increasing from $\eta$ through the largest attainable second-period precision $x_L$. Both continuation gains therefore exceed $D(\eta)>4c$. The direct conditions in Section~\ref{sec:choice} are weaker; \eqref{eq:sufficient} is convenient but not necessary.

\subsection{The evidence requirement and assessment}

Integrating customer stakes above the purchasing cutoff gives $3/8$. Replacing the revenue coefficient $1/4$ in the format comparison by $3/8$ therefore gives \eqref{eq:welfarediff}. Its zero is $\beta_S$. Because $c>0$, $\beta_S-\beta_P=4c/(3\Delta H)>0$. Under the interiority assumption of Proposition~\ref{prop:standard}, this proves the three regions and the welfare effect of the requirement. If $\beta_S\ge1$, the region in which the regulator strictly prefers basic evidence does not occur for $\beta<1$.

With $b$ observed, $[(1+q)m-b]/q=s$, a signal of precision $q$. Its value is $q/(1+q)$. Condition~\eqref{eq:audithigh} then makes intensive evidence optimal at each date. Relative to the unregulated basic--intensive sequence, discounted information value increases by $Z$, production cost rises by $c$, and assessment costs $K$. Private and social gains are $Z/4-c-K$ and $3Z/8-c-K$, respectively, proving Proposition~\ref{prop:audit}. The private gross gain is positive: even retaining the original evidence sequence, revealing $b$ strictly raises information value at both dates without increasing production cost. Optimising the audited sequence cannot do worse.

\section{General customer distributions}\label{app:general}

Let stakes have a continuous distribution $F$ on $[0,\bar\gamma]$, with $S(g)=1-F(g)$. Assume that $gS(g)$ has a unique interior maximiser $g^*$, and set
\begin{equation}\label{eq:MA}
 M=g^*S(g^*),\qquad A=\int_{g^*}^{\bar\gamma}\gamma\,dF(\gamma).
\end{equation}
Here $M$ is the revenue from selling a unit-value product at its optimal price. The quantity $A$ is the gross decision benefit generated for its buyers. These are different coefficients, since some buyers receive surplus.

For value $z$, writing the price as $p=gz$ gives revenue $zgS(g)$. Hence the optimal price is $g^*z$, revenue is $Mz$, and the same customer set buys at both dates in the public-history model. With a nondegenerate set of customers strictly above $g^*$,
\begin{equation}
 A-M=\int_{g^*}^{\bar\gamma}(\gamma-g^*)\,dF(\gamma)>0.
\end{equation}
For uniform stakes, $g^*=1/2$, $M=1/4$, and $A=3/8$, recovering the main text.

Replace $1/4$ by $M$ in the profit calculation and $3/8$ by $A$ in the welfare calculation. Whenever the second format is intensive after either first format, the switching thresholds become
\begin{equation}\label{eq:generalbeta}
 \beta_P=\frac{\Delta V-c/M}{\Delta H},\qquad
 \beta_S=\frac{\Delta V-c/A}{\Delta H}.
\end{equation}
When both are interior and $c>0$, their difference is
\begin{equation}
 \beta_S-\beta_P=\frac{c}{\Delta H}\left(\frac1M-\frac1A\right)>0.
\end{equation}
Thus the timing of the private and social switch does not rely on uniform demand. Under the corresponding audited-format condition $M[\Vhat(q_H)-\Vhat(q_L)]>c$, the assessment thresholds are $K_P=MZ-c$ and $K_S=AZ-c$. Their difference is $(A-M)Z$.

If a planner can also remove the access restriction, define total customer stakes
\begin{equation}
 B=\int_0^{\bar\gamma}\gamma\,dF(\gamma).
\end{equation}
With zero distribution cost, universal access is efficient conditional on producing the forecast. The planner's format threshold is then
\begin{equation}
 \beta_{FA}=\frac{\Delta V-c/B}{\Delta H}.
\end{equation}
The private strict preference for intensive second-period evidence implies that this planner also selects it after both histories, because $B>M$. If there are positive stakes below $g^*$, $B>A>M$, giving $\beta_P<\beta_S<\beta_{FA}$ for interior thresholds. This is a comparison of allocations and real costs; it presumes that the planner can finance production and is not restricted to monopoly subscription revenue.

\section{Continuous evidence choices}\label{app:continuous}

Suppose that the seller can choose any $q$ in a compact interval $Q=[\underline q,\overline q]\subset(0,\infty)$ and pays continuous cost $C(q)$. Use the general revenue coefficient $M$ from Appendix~\ref{app:general}. Assume that an active product is profitable at the relevant beliefs, so that operation rather than shutdown is optimal. Define the second-period profit function
\begin{equation}\label{eq:Wcont}
 \mathcal W(x)=\max_{q\in Q}\{MV(x,q)-C(q)\}.
\end{equation}
Since $V_x>0$, $\mathcal W$ is increasing. The first-period problem is
\begin{equation}\label{eq:continuousproblem}
 \max_{q\in Q}\left\{MV(\eta,q)-C(q)+\beta\mathcal W(\eta+\lambda(q))\right\}.
\end{equation}
Let $q^M$ be the unique maximiser of the first two terms, the evidence choice of a seller that disregards future sales.

\begin{proposition}\label{prop:continuous}
The smallest and largest optimal choices in \eqref{eq:continuousproblem} are nonincreasing in $\beta$, and every optimum is at most $q^M$. If $q^M$ is interior, $C$ is differentiable there, and the continuation maximiser is unique near $x=\eta+\lambda(q^M)$, every optimum is strictly below $q^M$ for $\beta>0$.
\end{proposition}

\begin{proof}
Write the objective as $G(q)+\beta J(q)$, where $G(q)=MV(\eta,q)-C(q)$ and $J(q)=\mathcal W(\eta+\lambda(q))$. Since $\lambda$ decreases and $\mathcal W$ increases, $J$ decreases. A higher $\beta$ therefore increases the relative attraction of every lower input compared with every higher input. The extreme optimal selections are nonincreasing. More explicitly, for $\beta'>\beta$ and corresponding optima $q'$ and $q$, the two optimality inequalities imply $(\beta'-\beta)[J(q')-J(q)]\ge0$, excluding $q'>q$ when $J$ is strictly decreasing.

For any $q>q^M$, uniqueness of $q^M$ implies $G(q)<G(q^M)$ and monotonicity implies $J(q)<J(q^M)$. Such a $q$ cannot maximise the dynamic objective. For strictness at $q^M$, the envelope theorem gives
\begin{equation}
 J'(q^M)=M V_x(\eta+\lambda(q^M),q_2^*)\lambda'(q^M)<0.
\end{equation}
The interior myopic optimum satisfies $G'(q^M)=0$. The derivative of the dynamic objective at that point is therefore $\beta J'(q^M)<0$ for $\beta>0$. A slightly lower input improves profit, excluding $q^M$ itself.
\end{proof}

The preceding result compares the first input with a static benchmark; it does not yet compare the two dates. For that comparison,
\begin{equation}\label{eq:cross}
 V_{qx}(x,q)=\frac{2q(1-xq^2)}{[1+xq(1+q)]^3}.
\end{equation}
Evidence and familiarity are complements only where $xq^2<1$. Suppose that the active static optimum $q^M(x)$ is unique for each relevant $x$, that the strict conditions of Proposition~\ref{prop:continuous} hold, and that
\begin{equation}\label{eq:rectangle}
 [\eta+\lambda(\underline q)]\overline q^{2}<1.
\end{equation}
Then the cross-partial is positive throughout the relevant rectangle. The static objective has increasing differences in $(x,q)$, implying that $q^M(x)$ is nondecreasing in $x$. Consequently,
\begin{equation}
 q_1<q^M(\eta)\le q^M(\eta+\lambda(q_1))=q_2.
\end{equation}
This proves a continuous-choice version of the increasing evidence sequence. Condition~\eqref{eq:rectangle} is sufficient and deliberately stronger than needed. The sign change in \eqref{eq:cross} explains why we avoid a global claim about the effect of familiarity or noisy calibration on the evidence input.

\section{Private histories and forward-looking customers}\label{app:private}

This appendix asks whether private learning by subscribers necessarily leads to an initial price discount beyond the static price. Suppose that only first-period buyers see the first forecast, and that the commercial seller observes purchase histories but not the forecast itself or its raw input. This latter restriction prevents a subsequent price or format from disclosing the private report. The seller still uses the same production technology; customers know that its precision choice is based on sales and common primitives. This is an alternative information structure, not an additional assumption about the public-history benchmark.

Fix a pair of formats, and write
\begin{equation}
 v=V(\eta,q_1),\qquad H=V(\eta+\lambda(q_1),q_2),\qquad L=V(\eta,q_2).
\end{equation}
The current value is $v>0$. At date 2 it is $H$ for experienced customers and $L$ for inexperienced ones, with $H>L>0$. The value does not depend on the private posterior mean. There is one second-period price, which cannot discriminate by purchase history. Customers discount the second date by $\delta$, the seller by $\beta$, and $0\le\delta\le\beta$. Equal discount factors satisfy this restriction.

\begin{proposition}\label{prop:private}
For the general stakes distribution of Appendix~\ref{app:general}, discounted revenue is bounded above by $M(v+\beta H)$. The bound is achieved with the initial cutoff $g^*$, prices $p_1=vg^*$ and $p_2=Hg^*$, and renewal by precisely the initial buyers. The second price is sequentially optimal at that cutoff.
\end{proposition}

\begin{proof}
Given an initial cutoff $g$ and a second price $p$, the marginal customer's continuation benefit from becoming experienced is
\begin{equation}
 \Delta(g,p)=[gH-p]_+-[gL-p]_+.
\end{equation}
The benefit of initial purchase is increasing in stakes, so cutoff demand is optimal. A first price supporting cutoff $g$ satisfies $p_1=gv+\delta\Delta(g,p)$. Its current-value revenue component $vgS(g)$ is at most $vM$. The remaining component is the capitalised continuation benefit. As $\Delta(g,p)\ge0$ and $\delta\le\beta$, it is enough to bound $\Delta(g,p)S(g)$ plus second-period revenue by $HM$.

If $p\le Lg$, all experienced customers renew and inexperienced customers with stakes at least $p/L$ also buy. The sum in question is
\begin{equation}
 g(H-L)S(g)+pS(p/L)\le(H-L)M+LM=HM.
\end{equation}
If $Lg\le p\le Hg$, all experienced customers renew and no inexperienced customer buys. The sum is
\begin{equation}
 (gH-p)S(g)+pS(g)=HgS(g)\le HM.
\end{equation}
If $p\ge Hg$, the marginal initial customer receives no continuation rent. Only experienced customers with stakes at least $p/H$ buy, giving
\begin{equation}
 pS(p/H)\le HM.
\end{equation}
These cases prove the revenue bound.

At $g=g^*$ and $p=Hg^*$, the marginal customer has no continuation surplus, so $p_1=vg^*$. The two revenues are $vM$ and $HM$, attaining the bound. Sequential optimality of $p_2$ follows because every customer's willingness to pay is at most $\gamma H$, so any second price yields at most $HM$; the proposed price attains that amount. Initial purchase utility is strictly increasing in $\gamma$ because $v>0$, and is zero at $g^*$, confirming the purchase decisions.
\end{proof}

The same argument permits endogenous second-period evidence. At cutoff $g^*$, every possible $q_2$ yields maximal revenue $M V(\eta+\lambda(q_1),q_2)$ by charging the corresponding $g^*H$. Thus the sequentially optimal continuation input maximises this revenue net of cost. The upper bound applies to every alternative cutoff and continuation input, and the static cutoff attains the resulting optimal two-period objective. The evidence problem is therefore the same as in the public-history model.

This extension uses stable one-dimensional stakes, one price at date 2, no additional public disclosure of the private report, and $\delta\le\beta$. It does not rule out learning-induced penetration pricing outside those restrictions. In particular, the proof does not extend to customers that discount the future less than the seller, or to arbitrary learning about match values.

\section{State feedback and changes in the source}\label{app:feedback}

\subsection{Additional information about the first state}

Suppose that, after the first decision but before the next sale, a public signal $r=\theta+u$ becomes available, where $u\sim N(0,\rho^{-1})$ is independent of the original evidence. The residual variance of $\theta$ conditional on $r$ is $1/(1+\rho)$. In \eqref{eq:bsignal}, the residual noise variance for learning about $b$ is therefore
\begin{equation}
 q+\frac{q^2}{1+\rho}.
\end{equation}
The precision increment becomes
\begin{equation}\label{eq:feedback}
 \lambda(q;\rho)=\frac{1+\rho}{q(1+\rho+q)},\qquad
 \frac{\partial\lambda}{\partial q}
 =-\frac{(1+\rho)(1+\rho+2q)}{q^2(1+\rho+q)^2}<0.
\end{equation}
At $\rho=0$ this is the baseline increment. When the state is fully revealed, the limit is $1/q$. Feedback strengthens learning about the component but preserves the fact that a basic forecast supplies more of that learning under our reporting rule. The relevant continuation values and evidence conditions must, of course, be recalculated under the new increment.

\subsection{An imperfectly persistent component}

Let $b_1$ have the initial distribution in \eqref{eq:prior}, and let
\begin{equation}
 b_2=\varrho b_1+\nu,\qquad \nu\sim N(0,\sigma_\nu^2),\quad \sigma_\nu^2>0,
\end{equation}
with independent innovation $\nu$. A first report of precision $q_1$ gives posterior variance $1/[\eta+\lambda(q_1)]$ about $b_1$. The induced precision about $b_2$ is
\begin{equation}
 x_2(q_1)=\left[\frac{\varrho^2}{\eta+\lambda(q_1)}+\sigma_\nu^2\right]^{-1}.
\end{equation}
For $\varrho\ne0$, this precision decreases with $q_1$. Lower initial evidence still improves knowledge relevant for using the second forecast. When $\varrho=0$, the two components are independent and this channel vanishes. For fixed positive innovation variance, it disappears continuously as $\varrho\to0$. This calculation establishes survival of the learning channel, not a global monotonicity result for the seller's optimal input in $\varrho$.

\section{Assessment versus an evidence requirement}\label{app:assessment}

The threshold $K_S$ in Proposition~\ref{prop:audit} compares an assessment with the seller's unregulated basic--intensive sequence. If the regulator can instead require intensive evidence without an assessment, the best unaudited outcome may be the intensive--intensive sequence. Its welfare advantage over the unregulated outcome is
\begin{equation}
 T=\frac38(\Delta V-\beta\Delta H)-c.
\end{equation}
An assessment is better than the best of these two unaudited outcomes exactly when
\begin{equation}\label{eq:Kbest}
 K<K_S-\max\{0,T\}.
\end{equation}
The audited intensive format is also preferred by this constrained regulator: its value increment is weighted by $3/8$, exceeding the private weight $1/4$ under which it is already optimal. Equation~\eqref{eq:Kbest} therefore compares the best formats with and without assessment, while maintaining monopoly access. It makes explicit why the private underadoption interval does not, by itself, select the best policy instrument.

\section{Reproducing the numerical example}\label{app:numeric}

All numerical comparisons use $\eta=0.1$, $q_L=0.5$, $q_H=1$, $C(q_L)=0$, and $C(q_H)=0.008$. Substitution into \eqref{eq:V} and \eqref{eq:lambda} gives
\begin{align}
 V(0.1,0.5)&=0.02325581,&V(0.1,1)&=0.08333333,\\
 x_L&=1.43333333,&x_H&=0.60000000,\\
 V(x_L,1)&=0.37068966,&V(x_H,1)&=0.27272727,\\
 \Delta V&=0.06007752,&\Delta H&=0.09796238.
\end{align}
The second-period intensive-minus-basic value gains are $0.19799889$ after $q_L$ and $0.16927900$ after $q_H$. Both exceed $4c=0.032$. The thresholds are
\begin{equation}
 \beta_P=0.28662,\qquad \beta_S=0.39550,\qquad \beta_{FA}=0.44994.
\end{equation}
At $\beta=0.34$, \eqref{eq:Z}--\eqref{eq:KS} give $Z=0.52071$, $K_P=0.12218$, and $K_S=0.18727$. Requiring intensive initial evidence without assessment improves welfare by about $0.00204$ at these parameters. Hence the assessment threshold relative to the best unaudited format is about $0.18523$, below $K_S$.

The model is solved analytically by normal updating and backward induction; the example does not estimate parameters from data. The accompanying Python script reproduces the example from the displayed formulas using rational arithmetic and compares all four possible format sequences. It also performs numerical checks of the private-history revenue bound and the switching conditions. These calculations check the example and help detect errors; the analytical proofs establish the general results.


\begin{thebibliography}{99}
\small
\bibitem[Bergemann and V\"alim\"aki(2006)]{bergemannvalimaki2006}
Bergemann, D., and V\"alim\"aki, J. (2006). Dynamic pricing of new experience goods. \emph{Journal of Political Economy}, 114(4), 713--743. \href{https://doi.org/10.1086/506923}{doi:10.1086/506923}.

\bibitem[Bonatti(2011)]{bonatti2011}
Bonatti, A. (2011). Menu pricing and learning. \emph{American Economic Journal: Microeconomics}, 3(3), 124--163. \href{https://doi.org/10.1257/mic.3.3.124}{doi:10.1257/mic.3.3.124}.

\bibitem[Campbell(2015)]{campbell2015}
Campbell, J. D. (2015). Ownership and pricing of information: A model and application to open access. \emph{Information Economics and Policy}, 33, 29--42. \href{https://doi.org/10.1016/j.infoecopol.2015.10.001}{doi:10.1016/j.infoecopol.2015.10.001}.

\bibitem[Galperti and Trevino(2020)]{galpertitrevino2020}
Galperti, S., and Trevino, I. (2020). Coordination motives and competition for attention in information markets. \emph{Journal of Economic Theory}, 188, 105039. \href{https://doi.org/10.1016/j.jet.2020.105039}{doi:10.1016/j.jet.2020.105039}.

\bibitem[Gentzkow and Shapiro(2006)]{gentzkowshapiro2006}
Gentzkow, M., and Shapiro, J. M. (2006). Media bias and reputation. \emph{Journal of Political Economy}, 114(2), 280--316. \href{https://doi.org/10.1086/499414}{doi:10.1086/499414}.

\bibitem[Gentzkow et~al.(2025)Gentzkow, Wong, and Zhang]{gentzkowwongzhang2025}
Gentzkow, M., Wong, M. B., and Zhang, A. T. (2025). Ideological bias and trust in information sources. \emph{American Economic Journal: Microeconomics}, 17(2), 162--213. \href{https://doi.org/10.1257/mic.20210406}{doi:10.1257/mic.20210406}.

\bibitem[Goering(1986)]{goering1986}
Goering, P. A. (1986). Learning, quality and prices. \emph{Information Economics and Policy}, 2(1), 23--47. \href{https://doi.org/10.1016/0167-6245(86)90013-2}{doi:10.1016/0167-6245(86)90013-2}.

\bibitem[Jann and Schottm\"uller(2026)]{jannschottmuller2026}
Jann, O., and Schottm\"uller, C. (2026). Manufactured ignorance: Strategic fragmentation in competitive news markets. Working paper, 19 June. \url{https://www.olejann.net/wp-content/uploads/hotelling.pdf}.

\bibitem[Kawamura and Le Quement(2023)]{kawamuralequement2023}
Kawamura, K., and Le Quement, M. T. (2023). News credibility and the quest for clicks. \emph{Journal of Public Economics}, 227, 105005. \href{https://doi.org/10.1016/j.jpubeco.2023.105005}{doi:10.1016/j.jpubeco.2023.105005}.

\bibitem[Lukyanov and Pototskiy(2026)]{lukyanovpototskiy2026}
Lukyanov, G., and Pototskiy, I. (2026). The price of a familiar perspective. Working paper, arXiv:2609.31752, version 2, 29 September. \href{https://doi.org/10.48550/arXiv.2609.31752}{doi:10.48550/arXiv.2609.31752}.

\bibitem[Panova(2026)]{panova2026}
Panova, E. (2026). Media market structure and confirmatory news. \emph{The Journal of Industrial Economics}, 74(3), 384--391. \href{https://doi.org/10.1111/joie.70029}{doi:10.1111/joie.70029}.

\bibitem[Saghafian et~al.(2022)Saghafian, Tomlin, and Biller]{saghafiantomlinbiller2022}
Saghafian, S., Tomlin, B., and Biller, S. (2022). The Internet of Things and information fusion: Who talks to who? \emph{Manufacturing \& Service Operations Management}, 24(1), 333--351. \href{https://doi.org/10.1287/msom.2020.0958}{doi:10.1287/msom.2020.0958}.

\bibitem[Sethi and Yildiz(2016)]{sethiyildiz2016}
Sethi, R., and Yildiz, M. (2016). Communication with unknown perspectives. \emph{Econometrica}, 84(6), 2029--2069. \href{https://doi.org/10.3982/ECTA13320}{doi:10.3982/ECTA13320}.

\bibitem[Sethi and Yildiz(2026)]{sethiyildiz2026}
Sethi, R., and Yildiz, M. (2026). Culture and communication. \emph{Journal of Economic Theory}, 233, 106152. \href{https://doi.org/10.1016/j.jet.2026.106152}{doi:10.1016/j.jet.2026.106152}.

\bibitem[Shapiro(1983)]{shapiro1983}
Shapiro, C. (1983). Optimal pricing of experience goods. \emph{The Bell Journal of Economics}, 14(2), 497--507. \href{https://doi.org/10.2307/3003650}{doi:10.2307/3003650}.
\end{thebibliography}
\end{document}